\documentclass[11pt, letterpaper]{article}
\usepackage[utf8]{inputenc}
\usepackage[margin=0.7in]{geometry}
\usepackage{apacite}
\usepackage{parskip}
\usepackage{natbib}
\usepackage{url} 
\usepackage[hidelinks]{hyperref}
\usepackage{amsmath,amsfonts,amssymb,amsthm}
\usepackage{bbold,fullpage}

\usepackage{commath}
\usepackage{booktabs}
\usepackage{graphicx}
\usepackage{float} 

\usepackage{enumitem} 

\usepackage{graphicx}      
\usepackage{subcaption}    
\usepackage{lipsum}   
\usepackage[affil-it]{authblk}
\usepackage{listings} 
\usepackage{caption}
\usepackage{setspace}
\newtheorem{definition}{Definition} 

\newtheorem{proposition}{Proposition}
\newtheorem{corollary}{Corollary}

\newtheorem{remark}{Remark}

\usepackage{tikz}
\usetikzlibrary{positioning}
\usetikzlibrary{arrows.meta}

\title{\vspace{-3em} 
\textbf{How Proxy Race Distorts Regression-Based Fairness Audits}\vspace{-0.75em}}
\author[1]{Xi Xin}
\author[2]{Giles Hooker}
\author[1]{Fei Huang\footnote{Correspondence: Xi Xin, xi.xin@unsw.edu.au;  Giles Hooker, ghooker@wharton.upenn.edu;   Fei Huang, feihuang@unsw.edu.au.}}

\affil[1]{UNSW Sydney, 
School of Risk and Actuarial Studies}
\affil[2]{University of Pennsylvania, Wharton School
Department of Statistics and Data Science}

\date{}

\begin{document}
\onehalfspacing
\maketitle
\begin{abstract}
    Proxy-based race inference is increasingly used to conduct fairness assessments when protected-class data are unavailable or legally restricted–most prominently in U.S. fair-lending enforcement, and now explicitly contemplated in emerging insurance regulation, including Colorado’s draft SB21-169 testing framework and New York’s Insurance Circular Letter No. 7. Despite this growing regulatory relevance, little is known about how standard regression-based discrimination analyses behave when race is measured with error through proxies such as Bayesian Improved Surname Geocoding (BISG) or Bayesian Improved First Name and Surname Geocoding (BIFSG). This paper studies the consequences of using proxy-imputed race as a categorical regressor in regression-based fairness assessments. Treating proxy race as a categorical covariate subject to misclassification, we show that proxy-based coefficients become weighted mixtures of true group effects, systematically shrinking estimated disparities toward the majority group–even when overall classification accuracy is high. Empirically, using a linked North Carolina voter–insurance dataset with self-reported race and ZIP-level auto insurance premiums, we demonstrate two mechanisms through which it distorts inference: (i) the intrinsic mixing of group effects implied by misclassification, and (ii) structured errors that vary with ZIP-level racial composition and socioeconomic conditions and remain correlated with pricing residuals after controls. As a result, regression-based disparity estimates can be attenuated or amplified relative to analogous analyses based on self-reported race. Our findings caution against treating proxy race as a plug-in substitute in regulatory testing and highlight design implications for proxy-based audit frameworks in insurance and other high-stakes domains.

\end{abstract}

\section{Introduction}

With the rapid expansion of big data and artificial intelligence (AI), concerns about discrimination against minority groups have intensified across many high-stakes domains. Empirical tests for discrimination, however, are often difficult to conduct, as sensitive attributes such as race or gender are typically unavailable or legally restricted from collection. To address this limitation in the context of racial discrimination, proxy-based approaches to race imputation have emerged, inferring individuals’ racial or ethnic identities from auxiliary information such as individual names and residential location, and enabling fairness assessments when direct race information is unavailable.

Proxy-based approaches to race inference for fairness assessment have been used most prominently in fair-lending settings. Under the Equal Credit Opportunity Act (ECOA) and its implementing Regulation B, creditors are prohibited from discriminating on the basis of race, ethnicity, and other protected characteristics, while generally barred from collecting applicants’ race or ethnicity for non-mortgage credit products\footnote{Mortgage lending represents a notable exception: the Home Mortgage Disclosure Act explicitly authorizes and requires the collection and reporting of applicants’ race and ethnicity.}. This regulatory structure has created an environment in which proxy-inferred sensitive attributes have become a practical tool for assessing potential disparities when direct measures are unavailable, including in a series of prominent auto-lending enforcement actions initiated with the 2013 Ally Bank case and followed by related cases during 2014–2016 (CFPB, \citeyear{cfpb2013ally, cfpb2015honda, cfpb2015fifththird, cfpb2016toyota}). These applications generated substantial controversy, as critics questioned whether proxy-based racial imputations could reliably support actions with substantial legal and financial consequences, prompting sustained academic and policy debate over their appropriate use and limitations \citep{HouseFinancialServices2015CFPB,koren2016feds}.


More recently, fairness concerns have increasingly migrated into the insurance domain, driven in large part by insurers’ expanding use of external consumer data and AI-driven predictive models throughout the insurance lifecycle. In response, insurance regulators have emphasized the need for insurers to implement appropriate testing and validation procedures to identify the potential for unfair discrimination in the decisions and outcomes arising from predictive models and AI systems (NAIC, \citeyear{NAIC2023AIModelBulletin}). The use of proxy-based race inference has also been noted in recent actuarial and industry discussions \citep{AAA2023DataBias, SOAImputingRace2024}. Several emerging insurance regulatory frameworks explicitly contemplate that protected-class status may be determined using data available to insurers or reasonably inferred using accepted statistical methodologies, including proxy-based approaches, when  direct race or other protected-class information is unavailable. Together, these developments create a concrete institutional setting in which proxy-based race inference and regression-based disparity testing are increasingly used to operationalize quantitative fairness assessments in insurance. Appendix~\ref{app:regcontext} provides a detailed summary of these regulatory developments, including Colorado’s draft SB21-169 testing framework and New York’s Insurance Circular Letter No.~7.


Existing research has shown that proxy-based race inference can distort fairness assessments in important ways. Empirical studies using empirical mortgage data find that BISG-based regressions may either overestimate or underestimate racial disparities relative to analyses using self-reported race \citep{baines2014fair, zhang2018assessing}. In a closely related study, \citet{chen2019fairness} provide a formal analysis of proxy-based fairness assessment for group-average outcome disparities, deriving interpretable sufficient conditions under which the commonly used thresholded and weighted estimators over- or under-estimate the true demographic disparity. More recently, \citet{kallus2022assessing} show that when protected-class status is observed only through proxies, disparity measures are generally not point-identifiable, motivating a shift toward partial-identification frameworks that characterize the range of disparities consistent with the available data. Despite this growing literature, comparatively little formal analysis exists for a setting that mirrors common fair-lending regression practices and is increasingly contemplated in insurance regulatory frameworks: regression-based discrimination analyses in which race or ethnicity is measured with error through Bayesian Improved Surname Geocoding (BISG) or Bayesian Improved First Name and Surname Geocoding (BIFSG) proxies and used directly as a categorical regressor. Existing work has largely focused on (i) direct estimators of demographic disparity defined as differences in group-average outcomes, or (ii) proposing alternative estimators or deriving identification bounds when protected-class status is unobserved. In contrast, this paper focuses on this regression-based setting and examines the implications of using proxy-imputed race directly as a categorical regressor in discrimination analyses.

This paper studies how regression-based fairness assessments behave when race is measured with error through BISG or BIFSG proxies. Treating proxy-imputed race as a measurement of true (unobserved) race subject to misclassification, we characterize the resulting error structure via a confusion matrix and show that proxy-based regression coefficients become weighted mixtures of true group effects. We further establish sufficient conditions under which misclassification necessarily shrinks between-group variability in predicted outcomes. As a result, estimated disparities are systematically shrunk toward the majority group--even when overall classification accuracy appears high. We illustrate these mechanisms empirically using a linked North Carolina voter–insurance dataset that combines self-reported race with ZIP-level auto insurance premiums. The empirical analysis highlights two distinct channels through which proxy race distorts inference: (i) an intrinsic mixing of group effects implied by misclassification, and (ii) structured proxy misclassification errors that vary with ZIP-level racial composition and socioeconomic conditions and remain correlated with pricing residuals after controls. Accordingly, regression-based disparity estimates can be either attenuated or amplified relative to analyses based on self-reported race.

More fundamentally, when proxy race is used as a categorical regressor, regression coefficients no longer directly recover true group disparities. Instead, they reflect a confusion-matrix–weighted transformation of those disparities induced by misclassification. In this sense, proxy-based regression audits alter the estimand itself rather than merely introducing classical measurement noise. Our results provide a mechanism-based explanation for why proxy race distorts regression-based discrimination analysis, caution against treating proxy race as a neutral plug-in substitute in regulatory testing, and offer practical implications for the design and interpretation of proxy-based audit frameworks in regulatory and industry settings. 

While we illustrate the mechanism using BISG/BIFSG and insurance pricing, the analysis is not specific to this empirical setting. Any regression-based audit that replaces a categorical protected attribute with a systematically misclassified proxy inherits a confusion-matrix transformation of group effects. The distortion we characterize therefore applies broadly across protected attributes, proxy constructions, and domains in which regression-based fairness testing is employed.

The remainder of the paper is organized as follows. Section~\ref{sec:related} reviews related studies on BISG/BIFSG and examines how their measurement errors affect fairness analysis; Section~\ref{sec:errors-categorical} develops our regression framework with misclassified categorical covariates and analyzes how misclassification mixes group effects and shrinks between-group variation in predicted outcomes; Section~\ref{sec:empirical} presents the empirical analysis using the North Carolina matched data; Section~\ref{sec:policy} discusses policy implications for regulators and industry practitioners; and Section~\ref{sec:conclusion} concludes.

\section{Related Studies} \label{sec:related}

\subsection{BISG and BIFSG Algorithms}

Bayesian Improved Surname Geocoding (BISG) is a widely used method for imputing race or ethnicity when individual-level demographic information is unavailable or legally restricted \citep{elliott2008new,elliott2009using}. It formalizes earlier practices that relied on geocoding-only or surname-only inference by combining the two sources of information within a Bayesian framework \citep{fiscella2006use}. Let $P(r \mid s, g)$ denote the BISG probability that an individual with surname $s$ and geographic area of residence $g$ belongs to racial or ethnic group $r$. Here $r \in \{1,\dots,R\}$ indexes racial or ethnic groups, where $R$ denotes the total number of groups, and $r'$ indexes the summation over all groups. The symbols $s$, $f$, and $g$ denote an individual's surname, first name, and geographic area of residence, respectively. BISG combines surname information with the local racial composition via Bayes’ rule:
\begin{equation}
    P(r \mid s, g)
    =
    \frac{P(r \mid s)\, P(g \mid r)}
         {\sum_{r'=1}^{R} P(r' \mid s)\, P(g \mid r')}.
\end{equation}
Here, $P(r \mid s)$ denotes the surname-based probability (i.e., typically obtained from Census surname tabulations), and $P(g \mid r)$ denotes the geography-based probability (i.e., the proportion of race $r$ in location $g$ within the entire race $r$ population). Originally developed in the context of health care research on racial and ethnic disparities, its use has since been institutionalized in regulatory settings, most notably in fair lending and discrimination audits conducted by U.S. regulatory agencies such as the Consumer Financial Protection Bureau \citep{cfpb2014using}.

Bayesian Improved First Name and Surname Geocoding (BIFSG) extends BISG by incorporating first-name information \citep{voicu2018using}. Let $P(r \mid s, f, g)$ denote the estimated probability
that an individual with surname $s$, first name $f$, and geographic location $g$ belongs to racial or ethnic group $r$. The BIFSG estimator can be written as:

\begin{equation}
    P(r \mid s, f, g) = \frac{P(r \mid s) \, P(f \mid r) \, P(g \mid r)}{\sum_{r'=1}^{R} P(r' \mid s) \, P(f \mid r') \, P(g \mid r')}.
\end{equation}

Here, $P(f \mid r)$ denotes the first-name-based probability. A key assumption underlying both BISG and BIFSG is that, conditional on race or ethnicity, geolocation, surname, and first name are mutually independent, i.e.,
\[
P(g \mid r, s) = P(g \mid r),
\qquad
P(f \mid r, s, g) = P(f \mid r).
\]

\subsection{Limitations of BISG and BIFSG}

While BISG and BIFSG provide practical tools for imputing race or ethnicity, their use relies on strong assumptions and introduces structured errors that are particularly consequential in downstream regression-based analyses. First, the conditional-independence assumption required by BISG and BIFSG is quite strong and often implausible in practice. For example, individuals may live near relatives who share surnames, violating the assumption that geography and surname are independent conditional on race \citep{greengard2023improved}. Second, the residuals of BISG and BIFSG frequently correlate with outcomes of interest and with other explanatory variables used in regression models \citep{voicu2018using, chen2019fairness, mccartan2024estimating}. Third, proxy construction requires non-missing surname, first-name, and geographic inputs; unmatched or missing information can lead to data attrition or heterogeneous use of reduced-information proxies, introducing additional selection and composition concerns \citep{voicu2018using}. Despite these concerns, BISG and BIFSG remain popular in practice because they often achieve seemingly high classification accuracy across racial groups, which may create a misleading impression that downstream regression-based analyses built on these proxies are likewise reliable.


\subsection{Measurement Errors of BISG or BIFSG in Fairness Analysis}

Beyond the intrinsic limitations of BISG and BIFSG as race-imputation tools, a growing literature examines how BISG-based racial proxies affect downstream fairness and discrimination analyses in empirical settings where benchmark race measures are available for comparison. Using empirical mortgage data, \citet{baines2014fair} and \citet{zhang2018assessing} show that regression-based discrimination analyses based on BISG-imputed race tend to overestimate racial outcome disparities relative to analogous analyses based on self-reported race. However, subsequent empirical evidence indicates that bias associated with proxy race measures can operate in either direction: attenuation of racial disparities has been documented both in loan approval analyses using image-based race benchmarks \citep{greenwald2024regulatory} and in mortgage outcome regressions depending on how proxies are incorporated \citep{voicu2018using}.

In a closely related study, \citet{chen2019fairness} analyze demographic disparity, defined as the difference in group-average outcomes between advantaged and disadvantaged groups. They provide interpretable sufficient conditions under which two common estimators -- the thresholded and weighted estimators -- over- or under-estimate that disparity. They decompose bias into inter-geolocation components (e.g., correlations between location and socioeconomic status in racially homogeneous areas) and intra-geolocation components (e.g., within-location patterns such as affirmative action favoring minority groups). The overestimation often observed in BISG-based empirical studies is consistent with inter-geolocation effects dominating the bias. The paper also shows that disparity estimates based on imputed race are highly sensitive to arbitrary tuning choices, such as the imputation threshold. \citet{chen2019fairness} focus on disparity estimation for group-average outcomes under proxy imputation; in contrast, our work considers regression-based discrimination analyses in which proxy-imputed race is used in place of true race.

In parallel, a growing literature addresses how to correct or mitigate bias when estimating outcome disparities using BISG-based racial proxies. \citet{kallus2022assessing} formulate fairness assessment with an unobserved protected class as a data combination problem. They show that commonly used disparity measures are not point-identifiable from proxy information alone and provide sharp partial-identification sets, with optimization procedures to compute and visualize these bounds. Relatedly, \citet{elzayn2025measuring} also adopt a partial-identification approach to estimate audit-rate disparities in the tax-audit setting. \citet{mccartan2025estimating} propose Bayesian Instrumental Regression for Disparity Estimation (BIRDiE), which uses surname as an instrument for race and BISG probabilities as inputs to obtain bias-corrected disparity estimates. They show that BISG plug-in estimates often attenuate disparities because BISG residuals correlate with outcomes; BIRDiE corrects this under the assumption that surname is conditionally independent of the outcome given (unobserved) race, location, and other observed covariates. 


Methodologically, we treat proxy race labels from BISG/BIFSG as measurements of true (unobserved) race subject to misclassification and summarize the resulting error structure using a confusion matrix. In this sense, our framework is related to the classical literature on regression with misclassified covariates in epidemiology \citep{shaw2020stratos, keogh2020stratos}. However, rather than proposing bias-correction or adjustment procedures, we focus on analyzing the implications of directly using proxy race as a categorical regressor in discrimination analyses. This framework allows us to show how proxy-based regression coefficients become weighted mixtures of true group effects and to establish sufficient conditions under which the use of racial proxies necessarily shrinks between-group variability in predicted outcomes. In contrast to prior work that emphasizes estimation error or uncertainty quantification, our results provide a mechanism-based explanation for why proxy race can distort regression-based discrimination analysis, even when individual-level classification accuracy appears high.

\section{Errors in Categorical Variables} \label{sec:errors-categorical}

We develop a general regression framework that isolates the mechanism through which misclassification in a categorical regressor distorts regression-based fairness audits. The analysis is fully abstract and does not depend on any specific application. Instead, we consider a setting in which a categorical regressor—representing a protected attribute—is observed with systematic misclassification and used directly in regression.

The key intuition is simple: misclassification moves observations between categories, altering the composition of the observed groups used for estimation. Each observed (proxy) group therefore reflects not only its own members but also individuals whose true categories differ. As a result, the empirical information associated with each observed category reflects contributions from multiple underlying true categories. The framework below formalizes how this reallocation affects coefficient estimation.


Consider a simple regression framework with a single categorical feature represented by a dummy-variable matrix $X \in \mathbb{R}^{n \times p}$ and associated response $\mathbf{y} \in \mathbb{R}^n$, under the usual linear-model assumption
\begin{equation}
    \mathbf{y} = X \boldsymbol{\beta} + \boldsymbol{\epsilon},
    \label{eq:linear-model}
\end{equation}
where $\boldsymbol{\beta} \in \mathbb{R}^p$ denotes the vector of category-specific coefficients, and $\boldsymbol{\epsilon} \in \mathbb{R}^n$ is the error term satisfying $\mathbb{E}[\boldsymbol{\epsilon} \mid X] = \mathbf{0}$ and $\operatorname{Var}(\boldsymbol{\epsilon} \mid X) = \sigma^2 I_n$. In many empirical settings, however, the true categorical labels $X$ (e.g., self-identified race) are unobserved. Instead, we only have access to an imputed or proxy version $\tilde{X}$, which is contaminated by misclassification errors:
\begin{equation}
    \tilde{X} = X + E,
    \label{eq:misclass}
\end{equation}
where $E$ captures the misclassification errors between the true and proxy labels. In the context of racial discrimination analysis, $X$ would represent the true race, while $\tilde{X}$ corresponds to race inferred from proxies such as Bayesian Improved Surname and Geocoding (BISG) or Bayesian Improved First Name Surname Geocoding (BIFSG).

\subsection{Estimation under Misclassified Covariates}

We next examine how estimation of $\boldsymbol{\beta}$ is affected when the categorical regressor is contaminated by misclassification. In particular, suppose that we estimate $\boldsymbol{\beta}$ using $\tilde{X}$ instead of $X$. The resulting estimator is given by
\begin{equation}
    \label{eq:beta-tilde}
    \tilde{\boldsymbol{\beta}} = \left( \tilde{X}^{\top} \tilde{X} \right)^{-1} \tilde{X}^{\top} \mathbf{y},
\end{equation}
where $\tilde{X} = X + E$ as in~\eqref{eq:misclass}. Substituting $\mathbf{y} = X \boldsymbol{\beta} + \boldsymbol{\epsilon}$ into~\eqref{eq:beta-tilde} yields
\begin{align}
\tilde{\boldsymbol{\beta}} 
&= \left( \tilde{X}^\top \tilde{X} \right)^{-1} \tilde{X}^\top \left( X\boldsymbol{\beta} + \boldsymbol{\epsilon} \right) \notag \\
&= \left( \tilde{X}^\top \tilde{X} \right)^{-1} \tilde{X}^\top \left( (\tilde{X} - E)\boldsymbol{\beta} + \boldsymbol{\epsilon} \right) \notag \\
&= \boldsymbol{\beta} 
- \left( \tilde{X}^\top \tilde{X} \right)^{-1} \tilde{X}^\top E\boldsymbol{\beta}
+ \left( \tilde{X}^\top \tilde{X} \right)^{-1} \tilde{X}^\top \boldsymbol{\epsilon}.
\label{eq:beta-tilde-decomp}
\end{align}
The decomposition in \eqref{eq:beta-tilde-decomp} separates the impact of misclassification from the usual sampling noise. Equivalently, we can rewrite \eqref{eq:beta-tilde-decomp} as
\begin{equation}
\label{eq:systematic-plus-noise}
\tilde{\boldsymbol{\beta}}
=
(\tilde{X}^{\top}\tilde{X})^{-1}\tilde{X}^{\top}X\boldsymbol{\beta}
\;+\;
(\tilde{X}^{\top}\tilde{X})^{-1}\tilde{X}^{\top}\boldsymbol{\epsilon}.
\end{equation}
The first term is the \emph{systematic (noise-free) component}: it is the deterministic mapping from the true group signal $X\boldsymbol{\beta}$ into the proxy categories induced solely by misclassification. Conditioning on $(X,\tilde X)$ and using $\mathbb{E}[\boldsymbol{\epsilon}\mid X]=\mathbf{0}$ yields
\begin{equation}
\label{eq:conditional-mean-beta-tilde}
\mathbb{E}[\tilde{\boldsymbol{\beta}}\mid X,\tilde X]
=
(\tilde{X}^{\top}\tilde{X})^{-1}\tilde{X}^{\top}X\boldsymbol{\beta},
\end{equation}
so the remaining term in \eqref{eq:systematic-plus-noise} captures sampling variation.

For comparison with the standard continuous-valued errors-in-variables decomposition and the associated attenuation intuition, see Appendix~\ref{app:continuous-me}. We formalize the decomposition of $\tilde{\boldsymbol{\beta}}$ below.

\begin{definition}[Decomposition of the Estimated Coefficient under Misclassification]
\label{def:decomposition}
Under the model \eqref{eq:linear-model} and the measurement-error structure \eqref{eq:misclass}, the estimator based on the proxy matrix $\tilde{X}$ can be decomposed as
\begin{equation}
\tilde{\boldsymbol{\beta}} = \boldsymbol{\beta}
     - (\tilde{X}^{\top} \tilde{X})^{-1} \tilde{X}^{\top} E \boldsymbol{\beta}
     + (\tilde{X}^{\top} \tilde{X})^{-1} \tilde{X}^{\top} \boldsymbol{\epsilon}.
\end{equation}
Here, $\boldsymbol{\beta}$ denotes the true coefficients, the term $-(\tilde{X}^{\top} \tilde{X})^{-1} \tilde{X}^{\top} E \boldsymbol{\beta}$ captures the bias due to misclassification, and the first two terms form the systematic (noise-free) component
$\mathbb{E}[\tilde{\boldsymbol{\beta}}\mid X,\tilde X]$. The term $(\tilde{X}^{\top} \tilde{X})^{-1} \tilde{X}^{\top} \boldsymbol{\epsilon}$ represents the stochastic variation due to sampling noise.
\end{definition}

\subsection{Category-wise Bias under Misclassification} \label{sec:category-bias}

We now focus on the structure of the bias term in Definition \ref{def:decomposition} when the categorical covariate is dummy-encoded. In this setting, each row of $X$ (and $\tilde{X}$) contains exactly one entry equal to 1, corresponding to the true (or predicted) category membership. Let $n_j$ denote the number of observations truly in category $j$, and let $\tilde{n}_j$ denote the number predicted as category $j$ under $\tilde{X}$. We also define $n_{k\to j}$ as the number of observations truly belonging to category $k$ but misclassified into category $j$, $n_{j\to j}$ as the number correctly classified as $j$, and $n_{\text{out},j}$ as the number of observations in category $j$ that are misclassified into some other category. Then, we have
\[
X^\top X = \operatorname{diag}(\mathbf{n}),
\qquad
\tilde{X}^\top \tilde{X} = \operatorname{diag}(\tilde{\mathbf{n}}).
\]
\paragraph{Derivation.}
The $j$-th component of the bias term in \eqref{eq:beta-tilde-decomp} can be written as
\begin{equation}
\label{eq:bias-component}
\big[-(\tilde{X}^\top\tilde{X})^{-1}\tilde{X}^\top E\boldsymbol{\beta}\big]_j
    = -\frac{1}{\tilde{n}_j}\sum_{i=1}^n \tilde{X}_{ij}\, (E_{i\cdot} \boldsymbol{\beta}),
\end{equation}
where $E_{i\cdot}$ denotes the $i$-th row of $E$, so that $E_{i\cdot} = \tilde{X}_{i\cdot} - X_{i\cdot}$. Here $X_{i\cdot}$ and $\tilde{X}_{i\cdot}$ denote the $i$-th rows of $X$ and $\tilde{X}$.
Since $E = \tilde{X} - X$ and each row of both matrices contains exactly one 1, we have
\[
E_{i\cdot} \boldsymbol{\beta}
    = (\tilde{X}_{i\cdot} - X_{i\cdot})\boldsymbol{\beta}
    = \beta_{\tilde{c}(i)} - \beta_{c(i)},
\]
where $\tilde{c}(i)$ and $c(i)$ denote the predicted and true categories of observation $i$, respectively. We now restrict the sum in \eqref{eq:bias-component} to those observations predicted as category $j$ (i.e., $\tilde{X}_{ij}=1$), since other observations contribute nothing to the summation:
\[
    \sum_{i=1}^n \tilde{X}_{ij}\,(E_{i\cdot} \boldsymbol{\beta})
    = \sum_{i:\tilde{c}(i)=j}(\beta_j - \beta_{c(i)}).
\]
We can decompose this summation into correctly classified and misclassified cases:
\[
\sum_{i:\tilde{c}(i)=j}(\beta_j - \beta_{c(i)})
    = n_{j\to j}(\beta_j - \beta_j)
      + \sum_{k\neq j} n_{k\to j}(\beta_j - \beta_k)
    = \sum_{k\neq j} n_{k\to j}(\beta_j - \beta_k).
\]
Substituting this into~\eqref{eq:bias-component} yields
\begin{equation*}
\label{eq:bias-intermediate}
\big[-(\tilde{X}^\top\tilde{X})^{-1}\tilde{X}^\top E\boldsymbol{\beta}\big]_j
    = \frac{1}{\tilde{n}_j}
      \left[
        -\Big(\sum_{k\neq j} n_{k\to j}\Big)\beta_j
        + \sum_{k\neq j} n_{k\to j}\beta_k
      \right].
\end{equation*}

Finally, noting that $n_{j\to j} = n_j - n_{\text{out},j}$ and that
\[
\tilde{n}_j = n_{j\to j} + \sum_{k\ne j} n_{k\to j}, 
\]
we arrive at the simplified expression
\begin{align*}
\big[-(\tilde{X}^\top\tilde{X})^{-1}\tilde{X}^\top E\boldsymbol{\beta}\big]_j
&= 
\frac{1}{\tilde{n}_j}
\left[
    -(\tilde{n}_j - n_{j\to j})\beta_j + \sum_{k\neq j} n_{k\to j}\beta_k
\right] \notag\\
&=
-\beta_j
+ \frac{1}{\tilde{n}_j}
  \left[
      n_{j\to j}\beta_j + \sum_{k\neq j} n_{k\to j}\beta_k
  \right] \notag\\
&=
-\beta_j
+ \frac{(n_j - n_{\text{out},j})\beta_j + \sum_{k\neq j} n_{k\to j}\beta_k}
       {\tilde{n}_j}.
\label{eq:bias-final}
\end{align*}

Intuitively, when individuals from one true category are systematically misclassified into another, the proxy category incorporates part of the signal from other true categories that are misclassified into it. The resulting coefficient is therefore a weighted mixture of group effects determined by the pattern of misclassification, as formalized below.

\begin{definition}[Category-wise Bias under Misclassification]
\label{def:category-bias}
For a dummy-coded categorical regressor, the bias in the estimated coefficient for category $j$ due to misclassification is given by
\begin{equation}
    \label{eq:category-bias-decomposition}
    \big[-(\tilde{X}^\top\tilde{X})^{-1}\tilde{X}^\top E\boldsymbol{\beta}\big]_j
=
-\beta_j + \frac{(n_j - n_{\text{out},j})\beta_j + \sum_{k\neq j} n_{k\to j}\beta_k}{\tilde{n}_j}.
\end{equation}
Equivalently, the systematic (noise-free) component of the proxy-based estimator admits the mixture representation
\begin{equation}
    \label{eq:mixture-representation}
    \mathbb{E}[\tilde{\beta}_j \mid X, \tilde{X}]
=
\frac{(n_j - n_{\text{out},j})\beta_j + \sum_{k\neq j} n_{k\to j}\beta_k}
     {\tilde{n}_j},
\end{equation}
which follows from $\mathbb{E}[\boldsymbol{\epsilon} \mid X] = \mathbf{0}$ under model \eqref{eq:linear-model}. This shows that $\tilde{\beta}_j$ is a weighted mixture of the true effect $\beta_j$ and the effects $\beta_k$ of other categories $k \neq j$ that are misclassified into category $j$. 
\end{definition}

\subsection{Expected Bias via the Confusion Matrix}

We now express the bias in terms of a confusion (misclassification) matrix that aggregates the flows between true and predicted categories. Let $C\in\mathbb{R}^{p\times p}$ be defined by
\[
C_{jk} = \Pr(\tilde{X}_{ij}=1 \mid X_{ik}=1),
\]
the probability that an observation from true class $k$ is predicted as class $j$. Let $\mathbf{n}=(n_1,\dots,n_p)^\top$ and $\tilde {\mathbf{n}}=(\tilde n_1,\dots,\tilde n_p)^\top$ denote the true and predicted class-count vectors, respectively, and let $\mathbf{n}\cdot\boldsymbol{\beta} := (n_1\beta_1,\dots,n_p\beta_p)^\top$ denote the elementwise product.

\begin{definition}[Expectation identities under a confusion matrix]
\label{def:confusion-identities}
With $C$ as above, the following expectation identities hold:
\begin{align}
\mathbb{E}\,[\tilde{\mathbf{n}} \mid \mathbf{n}] &= C\,\mathbf{n},
\label{eq:En-identity}
\\[2pt]
\mathbb{E}\,[\tilde{\mathbf{n}}\cdot \tilde{\boldsymbol{\beta}} \mid \mathbf{n}] &= C\,(\mathbf{n}\cdot \boldsymbol{\beta}).
\label{eq:Enb-identity}
\end{align}
\end{definition}

\begin{proof}
For \eqref{eq:En-identity}, write $n_{k\to j}=\sum_i \mathbf{1}\{X_{ik}=1,\ \tilde X_{ij}=1\}$ and $\tilde n_j=\sum_k n_{k\to j}$. 
Given $n_k$ items truly in class $k$, each is sent to class $j$ with probability $C_{jk}$, so 
$\mathbb{E}[n_{k\to j}\mid n_k]=C_{jk} n_k$. 
Summing over $k$ gives 
$\mathbb{E}[\tilde n_j\mid \mathbf{n}]=\sum_k \mathbb{E}[n_{k\to j}\mid \mathbf{n}]=\sum_k C_{jk} n_k=(C\,\mathbf{n})_j$. Stacking all $j$’s yields $\mathbb{E}\,[\tilde{\mathbf{n}} \mid \mathbf{n}] = C\,\mathbf{n}$.

From \eqref{eq:mixture-representation} and the derivation in Definition \ref{def:category-bias}, the systematic (noise-free) component satisfies 
\[
\tilde n_j\,\tilde\beta_j \;=\; \sum_{k=1}^p \beta_k\, n_{k\to j}.
\]
Taking the conditional expectation yields
\[
\mathbb{E}[\tilde n_j\,\tilde\beta_j \mid \mathbf{n}]
= \sum_k \beta_k\,\mathbb{E}[n_{k\to j}\mid \mathbf{n}]
= \sum_k \beta_k\, C_{jk}\, n_k
= \big(C\,(\mathbf{n}\cdot\boldsymbol{\beta})\big)_j,
\]
and stacking over $j$ yields \eqref{eq:Enb-identity}.

The identities in Definition~\ref{def:confusion-identities} show that, in expectation, the confusion matrix $C$ redistributes both (i) class counts and (ii) the ``signal mass'' $n_k\beta_k$ from true to predicted classes.
\end{proof}

\begin{proposition}[Approximate expected estimator under misclassification]
\label{prop:ratio-of-expectations}
Under a first-order (ratio-of-expectations) approximation,
\begin{equation}
\label{eq:roe}
\mathbb{E}\,\tilde{\boldsymbol{\beta}}
\;\approx\;
\operatorname{diag}\big(\mathbb{E}\,\tilde{\mathbf{n}}\big)^{-1}
\,\mathbb{E}\,\big(\tilde{\mathbf{n}}\cdot\tilde{\boldsymbol{\beta}}\big)
\;=\;
\operatorname{diag}(C\mathbf{n})^{-1}\,C\,(\mathbf{n}\cdot\boldsymbol{\beta}).
\end{equation}
\end{proposition}

\begin{proof}
By definition, $\tilde{\beta}_j = (\tilde{\mathbf{n}}\cdot\tilde{\boldsymbol{\beta}})_j/\tilde{n}_j$.
Taking expectations and applying a ratio-of-expectations (equivalently, first-order delta-method) approximation,
\[
\mathbb{E}\,\tilde{\beta}_j
\;\approx\;
\frac{\mathbb{E}\big[(\tilde{\mathbf{n}}\cdot\tilde{\boldsymbol{\beta}})_j\big]}{\mathbb{E}\,\tilde{n}_j}
=
\frac{\big(C(\mathbf{n}\cdot\boldsymbol{\beta})\big)_j}{(C\mathbf{n})_j},
\]
where we used \eqref{eq:En-identity}–\eqref{eq:Enb-identity}. Stacking over $j$ yields \eqref{eq:roe}.
\end{proof}

\begin{corollary}[Expected Bias in Matrix Form]
\label{thm:compact-bias}
The expected bias of the proxy-based estimator satisfies
\begin{equation}
\label{eq:general-bias}
\boldsymbol{\beta} - \mathbb{E}\,\tilde{\boldsymbol{\beta}}
\;\approx\;
\operatorname{diag}(C\mathbf{n})^{-1}
\Big(\operatorname{diag}(C\mathbf{n})\,\boldsymbol{\beta}
- C(\mathbf{n}\cdot\boldsymbol{\beta})\Big).
\end{equation}
Moreover, under the neutrality condition $C\mathbf{n}=\mathbf{n}$ (i.e., expected inflows and outflows balance for each class), \eqref{eq:general-bias} simplifies to
\begin{equation}
\label{eq:neutral-bias}
\boldsymbol{\beta} - \mathbb{E}\,\tilde{\boldsymbol{\beta}}
\;\approx\;
\operatorname{diag}(C\mathbf{n})^{-1}(I_p - C)(\mathbf{n}\cdot\boldsymbol{\beta}).
\end{equation}
\end{corollary}

\begin{proof}
Subtract \eqref{eq:roe} from $\boldsymbol{\beta}$:
\[
\boldsymbol{\beta} - \mathbb{E}\,\tilde{\boldsymbol{\beta}}
\;\approx\;
\boldsymbol{\beta}
- \operatorname{diag}(C\mathbf{n})^{-1}C(\mathbf{n}\cdot\boldsymbol{\beta})
=
\operatorname{diag}(C\mathbf{n})^{-1}\Big(\operatorname{diag}(C\mathbf{n})\,\boldsymbol{\beta}
- C(\mathbf{n}\cdot\boldsymbol{\beta})\Big),
\]
which is \eqref{eq:general-bias}. If $C\mathbf{n}=\mathbf{n}$, then $\operatorname{diag}(C\mathbf{n})\,\boldsymbol{\beta}=\operatorname{diag}(\mathbf{n})\,\boldsymbol{\beta}=\mathbf{n}\cdot\boldsymbol{\beta}$, and \eqref{eq:general-bias} reduces to \eqref{eq:neutral-bias}.
\end{proof}

\begin{remark}[On the ratio-of-expectations approximation]
The approximation in Proposition~\ref{prop:ratio-of-expectations}
corresponds to the leading term of a delta-method expansion for the
ratio $(\tilde{\mathbf{n}}\cdot\tilde{\boldsymbol{\beta}})_j / \tilde{n}_j$.
The neglected second-order terms are of order $O(1/n_j)$, where $n_j$ is
the sample size of category $j$. These terms vanish as $n_j$ becomes
large, so the ratio-of-expectations approximation~\eqref{eq:roe}
provides an accurate description when each class is reasonably well
represented in the data.
\end{remark}

\subsection{Variability of Predicted Values Under Misclassification}

In this subsection, we show that under neutral and reversible misclassification, the variability of predicted values necessarily decreases. We begin by stating the two key structural conditions.

\begin{definition}[Neutral Condition]
\label{def:neutral-condition}
The “neutral” condition means that, on average, each category keeps its total count:
\begin{equation}
    \mathbb{E}[\tilde n] = C\mathbf{n} = \mathbf{n}.
    \label{eq:neutral}
\end{equation}
That is, misclassification moves some observations between categories, but the expected total per class remains the same. Under this condition, the mean predicted value $\bar{y}$ stays unchanged, even if individuals move between categories.
\end{definition}

\begin{definition}[Detailed Balance Condition]
\label{def:detailed-balance}
A column-stochastic matrix $C$ is \emph{reversible} with respect to $\mathbf{n}$ if it satisfies the detailed-balance condition
\begin{equation}
    C_{jk}\, n_k = C_{kj}\, n_j
    \qquad (j,k = 1,\dots,p).
    \label{eq:detailed_balance}
\end{equation}
Equivalently, $C$ is self-adjoint in the weighted inner product
$\langle \mathbf{u},\mathbf{v}\rangle_n = \mathbf{u}^\top \operatorname{diag}(\mathbf{n})^{-1} \mathbf{v}$. Since $C_{jk} = n_{k\to j}/n_k$ is the empirical misclassification rate
from true class $k$ to predicted class $j$, the detailed balance
condition $n_k C_{jk} = n_j C_{kj}$ is equivalent to
\[
\mathbb{E}[n_{k\to j}] = \mathbb{E}[n_{j\to k}],
\]
meaning that in expectation, the number of misclassified individuals from $k$ to $j$ is the same as the number misclassified from $j$ to $k$. In the context of BISG-based racial imputation, this corresponds to the empirical symmetry that, for example, the number of Whites misclassified as Blacks is approximately the same as the number of Blacks misclassified as Whites.
\end{definition}

\begin{proposition}[Variance Shrinkage under Neutral and Detailed Balance Conditions]
\label{prop:shrinkage}
Suppose misclassification is neutral in the sense of Definition~\ref{def:neutral-condition} and reversible in the sense of Definition~\ref{def:detailed-balance}. Then the variability of the fitted values necessarily decreases:
\begin{equation}
    \sum_{i=1}^n (\tilde y_i -\bar y)^2
    \;\le\;
    \sum_{i=1}^n (\hat y_i -\bar y)^2,
    \label{prop:var_shrinkage}
\end{equation}
where $\bar{y}$ denotes the overall mean prediction. 
\end{proposition}

\begin{proof}
For dummy-coded predictors, all observations in category $j$ share the same fitted value $\beta_j$. Thus the total variability of predicted values can be written as a weighted variance of group means:
\begin{equation}
\sum_{i=1}^{n} (\hat{y}_i - \bar{y})^2
= \sum_{j=1}^{p} n_j(\beta_j - \bar{\beta})^2
= \big( (\mathbf{n}\!\cdot\!\boldsymbol{\beta}) - \bar{\beta} \mathbf{n} \big)^\top
\operatorname{diag}(\mathbf{n})^{-1}
\big( (\mathbf{n}\!\cdot\!\boldsymbol{\beta}) - \bar{\beta} \mathbf{n} \big),
\label{eq:origvar}
\end{equation}
where $\bar{\beta} = \frac{\sum_{j=1}^p n_j \beta_j}{\sum_{j=1}^p n_j}$. Define
\[
\mathbf{w} = (\mathbf{n}\!\cdot\!\boldsymbol{\beta}) - \bar{\beta}\mathbf{n},
\qquad
\|\mathbf{u}\|_n^2 := \mathbf{u}^\top \operatorname{diag}(\mathbf{n})^{-1}\mathbf{u},
\]
so that the variability under the true labels is simply
\[
\sum_{i=1}^n (\hat y_i - \bar y)^2 = \|\mathbf{w}\|_n^2.
\]
Under misclassification, the expected signal and expected class counts become $C(\mathbf{n}\cdot\boldsymbol{\beta})$ and $C\mathbf{n}$ (Definition~\ref{def:confusion-identities}). Applying the same representation as above and the ratio-of-expectations approximation gives
\[
\sum_{i=1}^n (\tilde y_i -\bar y)^2
\;\approx\;
\big( C(\mathbf{n}\!\cdot\!\boldsymbol{\beta}) - \bar{\beta} \, C\mathbf{n} \big)^\top
\operatorname{diag}(C\mathbf{n})^{-1}
\big( C(\mathbf{n}\!\cdot\!\boldsymbol{\beta}) - \bar{\beta} \, C\mathbf{n} \big).
\]

If misclassification is neutral, $C\mathbf{n}=\mathbf{n}$, so
\[
C(\mathbf{n}\!\cdot\!\boldsymbol{\beta}) - \bar\beta\,C\mathbf{n}
= C\big( (\mathbf{n}\!\cdot\!\boldsymbol{\beta})-\bar\beta \mathbf{n} \big)
= C\mathbf{w}.
\]
Hence the approximate variability after misclassification becomes
\[
\sum_{i=1}^n (\tilde y_i -\bar y)^2
\;\approx\;
\|C\mathbf{w}\|_n^2.
\]

To show contraction, consider the similarity transform
\[
M := \operatorname{diag}(\mathbf{n})^{-1/2}\, C \,\operatorname{diag}(\mathbf{n})^{1/2}.
\]

Under the detailed balance condition $C_{jk} n_k = C_{kj} n_j$, the matrix $M$ is symmetric and therefore admits an orthonormal eigenbasis with real eigenvalues $1=\lambda_1 \ge \lambda_2 \ge \cdots \ge \lambda_p \ge -1$. Define the standardized vector $\mathbf{z} := \operatorname{diag}(\mathbf{n})^{-1/2}\mathbf{w}$. Then
\[
\|\mathbf{w}\|_n^2 = \|\mathbf{z}\|_2^2,
\qquad
\|C\mathbf{w}\|_n^2 = \|M\mathbf{z}\|_2^2.
\]
Writing $\mathbf{z} = \sum_{\ell} a_\ell \mathbf{v}_\ell$ in the orthonormal eigenbasis $\{\mathbf{v}_\ell\}$ of $M$, we obtain
\[
\|C\mathbf{w}\|_n^2
= \|M\mathbf{z}\|_2^2
= \sum_{\ell=1}^p \lambda_\ell^2 a_\ell^2
\;\le\;
\sum_{\ell=1}^p a_\ell^2
= \|\mathbf{z}\|_2^2
= \|\mathbf{w}\|_n^2.
\]

Combining this with the representations above, we conclude that under neutral
and reversible misclassification,
\[
\sum_{i=1}^n (\tilde y_i -\bar y)^2
\;\le\;
\sum_{i=1}^n (\hat y_i -\bar y)^2,
\]
so misclassification necessarily shrinks the variability of the predicted values.

\end{proof}

\section{Empirical Studies} \label{sec:empirical}

\subsection{NC Voter--Insurance Matched Data}
Our empirical analysis uses a combined dataset constructed by linking publicly available voter registration records in North Carolina with zip-code--level auto insurance premium information. We focus on North Carolina because its voter file contains individual-level identifiers--in particular, full name and self-reported race/ethnicity--that are widely used in studies evaluating BISG-based racial proxies\footnote{The layout of the North Carolina voter file is publicly documented at \url{https://s3.amazonaws.com/dl.ncsbe.gov/data/layout_ncvoter.txt}.}.

Zip-code--level insurance premiums are obtained via web crawling from \textit{CarInsurance.com}, which reports average auto insurance rates produced by Quadrant Information Services. These premiums are calculated for a standardized insured profile: a driver with clean records, good credit, and state-minimum required liability coverage driving a 2017 Honda Accord\footnote{Because the premium measure is based on a fully standardized driver profile—with fixed credit quality, driving history, vehicle type, and coverage level—this construction is equivalent to intervening on these rating factors. In causal terms, these interventions remove all pathways through which race could influence premiums via creditworthiness, driving behavior, or vehicle choice. Hence, in our dataset, racial differences in the imputed premium can arise only through ZIP code, age, or gender, which are the only factors that vary across individuals.}. The website provides average monthly premiums for each combination of ZIP code, age group, and gender using data from up to six major insurance carriers. We collect these averages for the first 100 ZIP codes in North Carolina (ranked by voter population) for all combinations of age group, gender, and ZIP code. The age dimension in the premium data is reported as single-year ages from 16 to 24 and, for older individuals, at selected representative ages (25, 30, 35, \dots, 75).

We then merge these premium data with North Carolina’s individual-level voter registration records obtained from the state’s Board of Elections. For each individual, we assign an insurance premium corresponding to the average premium of their ZIP code $\times$ age group $\times$ gender cell. We refer to the resulting linked dataset as the NC Voter–Insurance Matched Data, which contains (i) individual demographic characteristics from the voter file; (ii) true race/ethnicity as self-reported in the voter records; and (iii) an imputed insurance premium representing the average cost faced by a standardized driver with similar demographic attributes in that ZIP code. Additional details on data sources and sample construction are provided in Appendix~\ref{app:datasets-used}.

\begin{table}[!ht]
\centering
\caption{Confusion matrix with row/column sums and per-row accuracy}
\label{tab:confusion}
\begin{tabular}{lrrrrrrr}
\toprule
 & \multicolumn{5}{c}{Reference} & Row Sum & Accuracy (\%) \\
\cmidrule(lr){2-6}
Prediction & Asian & Black & Hispanic & Others & White & & \\
\midrule
Asian    & 17907 &   448 &   380 &  6298 &    1815 &   26848  & 66.7 \\
Black    &  1836 & 280480 &  3983 & 11346 &   82036 &  379681 & 73.9 \\
Hispanic &  1087 &  2765 & 74506 &  6997 &   15313 &  100668 & 74.0 \\
Others   &  4089 & 10060 &  2496 &  7327 &   12737 &   36709 & 20.0 \\
White    &  6533 & 140072 & 12409 & 22671 & 1126267 & 1307952 & 86.1 \\
\midrule
Column Sum & 31452 & 433825 & 93774 & 54639 & 1238168 & 1851858 & \\
\bottomrule
\end{tabular}
\end{table}

The confusion matrix in Table~\ref{tab:confusion} is constructed using the BIFSG \emph{max-classification} rule, which assigns each individual to the racial group with the highest BIFSG posterior probability. The BIFSG probabilities are computed using the \texttt{wru} package \citep{khanna2022wru}, drawing on 2020 U.S.\ Census ZCTA-level data. The overall accuracy patterns are consistent with prior findings in the BISG/BIFSG literature: White and Black voters are classified with relatively high accuracy (86.1\% and 73.9\%, respectively), while the “Others’’ category exhibits substantial noise.

Importantly, the misclassification structure aligns with the methodological conditions discussed earlier. In particular, the flows between Black and White voters are comparatively balanced, producing an approximate symmetry that is relevant for the neutrality and detailed-balance conditions. 

\textbf{Socioeconomic Characteristics.} 
To incorporate ZIP-level socioeconomic status (SES) into the analysis, we supplement the matched data with ZIP Code Tabulation Area (ZCTA)--level socioeconomic indicators from the National Neighborhood Data Archive (NaNDA) \citep{nanda38528}. NaNDA provides ZCTA-level measures derived primarily from the American Community Survey (ACS) five-year estimates (2018–2022), including indicators capturing poverty, unemployment, income distribution, educational attainment, and neighborhood disadvantage and affluence. For our empirical analysis, we extract three SES variables -- median family income (MEDFAMINC), poverty rate (PPOV), and unemployment rate (PUNEMP) -- to account for socioeconomic variation across ZIP codes that may correlate with both racial composition and insurance pricing. These indicators capture socioeconomic variation relevant for assessing potential confounding channels in pricing and proxy-based discrimination analysis. Definitions of all variables used in Section~\ref{sec:empirical} are summarized in Appendix~\ref{app:variables-used}.

\subsection{Experiment 1: Regressing True Fitted Values on Imputed Race Labels}
\begin{table}[!t]
\centering
\caption{Experiment 1: Pricing Disparities Using Reported and BIFSG Race Labels}
\label{tab:exp1}
\renewcommand{\arraystretch}{1.05}

\text{Panel A: Raw Pricing Disparities} \\[0.2em]
\begin{tabular}{lccc}
\toprule
Race Group & Reported Race & Proxy-Induced Mixing & Proxy Race \\
\midrule
Asian      & 2.041 & 1.677 & 1.765 \\
Black      & 2.655 & 1.712 & 3.421 \\
Hispanic   & 3.546 & 2.524 & 3.323 \\
Others     & 2.533 & 1.330 & 2.928 \\
\bottomrule
\end{tabular}

\vspace{0.7em}

\text{Panel B: Adjusted Pricing Disparities} \\[0.2em]
\begin{tabular}{lcccc}
\toprule
 & \multicolumn{2}{c}{Reported Race} & \multicolumn{2}{c}{Proxy Race} \\
\cmidrule(lr){2-3} \cmidrule(lr){4-5}
Race Group & Control 1 & Control 2 & Control 1 & Control 2 \\
\midrule
Asian      & 1.399 & 1.180 & 1.385 & 1.183 \\
Black      & 2.550 & 1.736 & 3.263 & 2.272 \\
Hispanic   & 1.494 & 1.081 & 1.511 & 1.137 \\
Others     & 1.762 & 1.342 & 1.768 & 1.386 \\
\bottomrule
\end{tabular}

\begin{flushleft}
\footnotesize
Notes: Panel~A reports raw pricing disparities by race without controls. 
``Reported Race'' uses self-reported race from the NC voter file, which implements the regression \eqref{eq:exp1_eq1}. 
``Proxy-Induced Mixing'' implements the regression \eqref{eq:exp1_eq2}, 
which regresses fitted values from the reported-race model onto BIFSG-max labels, capturing the systematic mixing (noise-free) of true group effects induced by misclassification.
``Proxy Race'' reports disparities using BIFSG-max labels, which implements the regression \eqref{eq:exp1_eq3}.
Panel~B reports adjusted disparities from OLS regressions. 
``Control~1'' includes age-group and gender.
``Control~2'' additionally include  ZIP-level socioeconomic variables: median family income (MEDFAMINC), poverty rate (PPOV), and unemployment rate (PUNEMP). White is the reference group throughout, consistent with standard practice in regression-based disparity assessments.
\end{flushleft}
\end{table}

Our first experiment recovers the systematic (noise-free) component of the proxy-based regression coefficients characterized in Definition~\ref{def:decomposition}, isolating the intrinsic effect of replacing the true race labels $X$ with proxy labels $\tilde{X} = X + E$ in standard regression-based disparity analyses. As shown in our analytical derivations, misclassification acts on the true signal through the operator
\[
(\tilde{X}^{\top}\tilde{X})^{-1}\tilde{X}^{\top}X \boldsymbol{\beta},
\]
which redistributes the group effects $X\boldsymbol{\beta}$ across the proxy categories as shown in Definition~\ref{def:category-bias}.  To approximate this transformation empirically, we proceed in two steps.

\begin{enumerate}[leftmargin=1.5em]
    \item \textbf{Reported-Race Regression.}  
    We regress the insurance premium on the true race categories:
    \begin{align}
        y_\text{true} \sim \text{race}_{\text{true}},
        \label{eq:exp1_eq1}
    \end{align}
    obtaining coefficient estimates $\hat{\beta}$ representing disparities under the true labels.

    \item \textbf{Proxy-Induced Mixing Regression.}  
    We compute fitted values $\widehat{y}_{\text{true}} = X\hat{\beta}$, then regress these fitted values on
    the imputed race labels based on BIFSG-max:
    \begin{align}
        \widehat{y}_{\text{true}} \sim \text{race}_{\text{proxy}}.
        \label{eq:exp1_eq2}
    \end{align}
    This regression provides an empirical estimate of the transformation $(\tilde{X}^{\top}\tilde{X})^{-1}\tilde{X}^{\top}X \boldsymbol{\beta}$, 
    implemented using the sample analogue based on $\hat{\beta}$, and thus captures the systematic mixing of true group effects induced by misclassification, independent of sampling noise, as shown in Definition~\ref{def:decomposition}.

    \item \textbf{Proxy-Race Regression.}  
    Finally, we regress the insurance premium directly on the proxy race labels based on BIFSG-max:
    \begin{align}
        y_\text{true} \sim \text{race}_{\text{proxy}}.
        \label{eq:exp1_eq3}
    \end{align}
    This yields the full proxy-based regression coefficient, which combines systematic mixing and sampling variation, corresponding to the estimator characterized in Definition~\ref{def:decomposition}.

\end{enumerate}

Panel A of Table~\ref{tab:exp1} reports the coefficients from the true-race model and their proxy-based transformed counterparts. Across all racial groups, the mapping coefficients are substantially attenuated relative to the true-race coefficients. This pattern matches the mixture representation in Definition~\ref{def:category-bias} of Section~\ref{sec:category-bias}, under which each proxy coefficient averages together the true group effect $\beta_j$ with the effects of other groups that are misclassified into category $j$. In North Carolina, Whites constitute the majority population and also face the lowest average premiums in our merged dataset. As a result, White individuals misclassified into minority proxy groups mechanically pull the proxy-based coefficients for Asian, Black, Hispanic, and Others toward the White-group baseline. In other words, the large White group acts as an anchor in the mixture, mechanically compressing the minority-group disparities. The “BIFSG Race’’ column shows that once we reintroduce residual noise by regressing the observed premium directly on proxy race, some disparities -- particularly for Black and Others -- can be inflated relative to the true-race benchmark, even though the underlying misclassification structure is intrinsically attenuating.

Panel B of Table~\ref{tab:exp1} reports adjusted disparities after controlling for demographic and socioeconomic effects. Reporting conditional disparities is standard practice in regression-based discrimination analysis and is commonly required in regulatory testing frameworks (e.g., Colorado’s draft SB21-169 testing framework), where pricing differences are evaluated after accounting for approved rating factors such as age and gender. 

In addition to these pricing variables (Control~1), we also incorporate ZIP-level socioeconomic characteristics (Control~2). While these SES variables are not themselves direct pricing factors in our premium construction, they capture neighborhood-level conditions that may be correlated with both racial composition and insurance pricing. Including them allows us to assess whether racial disparities persist after adjusting for broader socioeconomic structure. 

As expected, adding controls reduces the magnitude of the reported-race disparities relative to the raw differences in Panel~A, suggesting that part of the observed pricing gaps can be statistically accounted for by differences in demographic composition and neighborhood socioeconomic conditions. However, the comparison between reported race and BIFSG race shows that proxy-based regressions can either attenuate or inflate disparities once controls are included. For example, the BIFSG-based Black coefficients remain larger than their reported-race counterparts under both sets of controls, while the Asian coefficients are more strongly compressed under BIFSG. These patterns illustrate that, in the presence of additional covariates and residual variation, the proxy-based regression no longer corresponds to a pure algebraic mapping of $X\hat{\beta}$. Taken together, Panel~A and Panel~B show that misclassification imposes a structural shrinkage toward the White-dominated mean, but once proxy race becomes correlated with the remaining pricing residuals after adjustment for observable controls, the resulting conditional disparities can be either attenuated or amplified relative to analyses based on self-reported race.

\subsection{Experiment 2: ZIP-Level Racial Deviation and Misclassification}

Experiment 1 isolates the intrinsic mixing of group effects induced by proxy misclassification, showing how misclassification redistributes true group disparities across proxy categories. However, the empirical results in Table~\ref{tab:exp1} also reveal that proxy-based regressions can produce disparities that deviate from the mixing-only benchmark--sometimes even amplifying disparities relative to analyses based on self-reported race. This indicates that proxy-induced mixing alone does not fully account for the observed distortion in regression-based fairness assessments, suggesting the presence of additional channels through which proxy race affects inference.

One plausible mechanism is that proxy misclassification is not random but systematically structured by local demographic and socioeconomic conditions. Intuitively, when a ZIP code is racially homogeneous--such as predominantly Black or predominantly White areas--individuals are more likely to be assigned to the locally predominant group by the proxy algorithm \citep{greenwald2024regulatory}. Recent empirical evidence further shows that BISG misclassification--including the magnitude of the resulting estimation error--is systematically structured by neighborhood socioeconomic characteristics and local racial composition \citep{argyle2024misclassification}. 

Experiment~2 therefore investigates whether structured proxy misclassification generates additional distortion channels beyond intrinsic mixing. Specifically, we examine (i) whether misclassification rates vary systematically with ZIP-level racial composition--particularly in more racially homogeneous areas--and (ii) whether misclassification patterns that remain after accounting for ZIP-level racial composition and socioeconomic heterogeneity continue to be correlated with pricing residuals.


Let $n_{ik}$ denote the number of voters of race $k$ in ZIP code $i$, and $\tilde{n}_{ik}$ the corresponding BIFSG-max counts. Let $n_k$ denote the statewide population of race $k$ and $n_i = \sum_k n_{ik}$ the population of ZIP $i$. To express all quantities on comparable ZIP-scale levels, we define the ZIP $\times$ race deviation and misclassification counts as
\[
d_{ik} := n_i\!\left(\frac{n_{ik}}{n_i} - \frac{n_k}{n}\right),
\qquad
r_{ik} := n_{ik} - \tilde{n}_{ik}.
\]
Thus, $d_{ik}$ measures how much ZIP $i$'s racial composition deviates (in counts) from statewide proportions, while $r_{ik}$ measures the direction and magnitude of BIFSG-max misclassification.

Let $\varepsilon_j$ be the individual residual from the reported-race model. The ZIP $\times$ race residual is defined as the sum of individual residuals:
\[
\epsilon_{ik} := \sum_{j \in (i,k)} \varepsilon_j,
\]
so that $d_{ik}$, $r_{ik}$, and $\epsilon_{ik}$ in the analysis are expressed on a comparable scale.

\begin{enumerate}[leftmargin=1.5em]

    \item \textbf{Misclassification–ZIP-Level Deviation Regression.}
    We estimate
    \begin{align}
        r_{ik} = \alpha_0 + \alpha_1 d_{ik} + \eta_{ik},
        \label{eq:exp2_eq1}
    \end{align}  
    where $\alpha_1$ captures whether misclassification increases in racially more homogeneous ZIP codes, i.e., those with larger $|d_{ik}|$.

    \item \textbf{Premium Residual–Misclassification Residual Regression.}
    We then assess whether the remaining misclassification residuals relate to pricing residuals:
    \begin{align}
        \epsilon_{ik} = \gamma_d d_{ik} + \gamma_e \eta_{ik} + \xi_{ik}.
        \label{eq:exp2_eq2}
    \end{align}
    Here, $d_{ik}$ acts as a confounder: ZIP codes with extreme racial compositions differ systematically in socioeconomic and risk characteristics, which affect both premiums and the scale of misclassification. The coefficient $\gamma_d$ therefore captures ZIP-level confounding, while $\gamma_e$ measures any remaining correlation between misclassification residuals and pricing residuals.

\end{enumerate}




\begin{table}[H]
\centering
\caption{Experiment 2: ZIP-Level Regressions for Misclassification and Premium Residuals}
\label{tab:exp2_diff}
\renewcommand{\arraystretch}{1.05}

\textbf{Panel A: Misclassification Regression ($r_{ik} \sim d_{ik}$)} \\[0.3em]
\begin{tabular}{lcccc}
\toprule
 & \multicolumn{2}{c}{White} & \multicolumn{2}{c}{Black} \\
\cmidrule(lr){2-3} \cmidrule(lr){4-5}
 & Baseline & + SES & Baseline & + SES \\
\midrule
Intercept 
 & 697.84*** & 878.97 & $-$541.44*** & -735.12      \\
Deviation ($d_i$) 
 & 0.1769*** & 0.1305*** & 0.2240*** & 0.1849*** \\
\midrule
Residual SE 
 & 781.9 & 726.7 & 711.0 & 692.6 \\
SES controls 
 & No & Yes & No & Yes \\
Observations 
 & 100 & 100 & 100 & 100 \\
\bottomrule
\end{tabular}

\vspace{0.8em}

\textbf{Panel B: Premium Residual Regression ($\epsilon_{ik} \sim d_{ik} + e_{ik}$)} \\[0.3em]
\begin{tabular}{lcccc}
\toprule
 & \multicolumn{2}{c}{White} & \multicolumn{2}{c}{Black} \\
\cmidrule(lr){2-3} \cmidrule(lr){4-5}
 & Baseline & + SES & Baseline & + SES \\
\midrule
Deviation ($d_i$) 
 & $-$6.655*** & -7.474***  & 3.473*** & 3.737***  \\
Misclass.\ Residual ($e_i$) 
 & 19.951** & 18.074**  & 7.313* &  8.776*  \\
\midrule
Residual SE 
 & 46380 &  45120 & 23870 & 23800  \\
SES controls 
 & No & Yes & No & Yes \\
Observations 
 & 100 & 100 & 100 & 100 \\
\bottomrule
\end{tabular}

\begin{flushleft}
\footnotesize
Notes: Panel~A corresponds to regression \eqref{eq:exp2_eq1},
$r_{ik} = \alpha_0 + \alpha_1 d_{ik} + \eta_{ik}$,
where $d_{ik}$ measures ZIP--race deviation from statewide racial composition and
$r_{ik}$ represents the BIFSG-max misclassification displacement. 
Panel~B implements regression \eqref{eq:exp2_eq2},
$\epsilon_{ik} = \gamma_d d_{ik} + \gamma_e \eta_{ik} + \xi_{ik}$,
where $\epsilon_{ik}$ denotes the ZIP $\times$ race residual sum from the reported-race premium model.
SES controls include ZIP-level median family income (MEDFAMINC), poverty rate (PPOV), and unemployment rate (PUNEMP).
Significance codes: * p$<$.05, ** p$<$.01, *** p$<$.001.
\end{flushleft}
\end{table}

Table \ref{tab:exp2_diff} summarizes the ZIP × race regressions for misclassification (Panel A) and premium residuals (Panel B). Panel A corresponds to \eqref{eq:exp2_eq1} and shows that BIFSG-max misclassification is strongly shaped by ZIP-level racial composition: the deviation coefficient $\alpha_1$ is positive and highly significant for both White and Black cohorts, indicating that misclassification increases in racially more homogeneous ZIP codes. Adding SES controls reduces the magnitude of $\alpha_1$ but leaves the main effect unchanged. Panel~B corresponds to \eqref{eq:exp2_eq2} and examines how these ZIP-level patterns relate to premium residuals. ZIP-level deviation ($d_{ik}$) is significant for both White and Black cohorts, highlighting that ZIPs with extreme racial compositions also differ systematically in socioeconomic and risk characteristics that affect premiums. After controlling for this structural component -- both via $d_{ik}$ and the SES variables -- the residual misclassification component ($\eta_{ik}$) remains significantly associated with ZIP-level premium residuals. The effect is larger for Whites, suggesting that proxy errors generate additional spurious ZIP-level correlations beyond those driven by racial heterogeneity or socioeconomic conditions. Together, the results indicate a two-step mechanism: (1) BIFSG misclassification is amplified in racially homogeneous ZIP codes, and (2) the resulting misclassification residuals introduce additional correlations with pricing residuals even after adjusting for ZIP-level racial and socioeconomic structure.

\section{Policy Implications for Regulators and Industry Practitioners} \label{sec:policy}

Our results have several implications for regulators and industry practitioners considering proxy-based quantitative testing frameworks.

First, regression-based fairness analyses that rely on proxy race require more caution than is often acknowledged in regulatory and industry discussions. In our empirical application, BISG/BIFSG misclassification produces a systematic shrinkage toward majority-group baselines, which can mask substantial underlying differences, particularly in markets where a single racial group constitutes a large majority of the population. At the same time, once proxy race becomes correlated with the pricing residuals that remain after adjusting for demographic and socioeconomic covariates, the resulting regression coefficients can either attenuate or inflate observed disparities relative to analyses based on self-reported race. These patterns illustrate that proxy race is not a neutral substitute for protected-class information, especially when regression outputs are used to support regulatory testing, supervisory assessments, or enforcement actions.

Second, proxy-based regression is not the only operational pathway for conducting quantitative fairness assessments. The use of proxy inference is better understood as a workaround driven by the absence of protected-class information -- a basic tension in which collecting the data needed to test for fairness is often restricted by fairness regulations. As highlighted in Section 4 of the SOA’s recent report \citep{SOAImputingRace2024}, direct demographic data collection remains possible under certain conditions, but it is frequently constrained by legal prohibitions, consumer reluctance, and concerns about how such information would be used. A complementary approach, emphasized in UK government work on demographic data access \citep{UKGov2023DemographicDataFairer}, is to use data intermediaries -- trusted third parties that enable demographic evaluation, validation, or secure linkage under privacy-preserving access frameworks without disclosing protected-class information to insurers. Such mechanisms offer a structurally different route to fairness assessment and, in many cases, a more transparent and statistically reliable alternative to unstructured reliance on proxy inference.

Third, our results point to a complementary line of research that seeks to improve proxy-based race inference itself. Our Experiment 2 in Section~\ref{sec:empirical} shows that BIFSG misclassification correlates strongly with ZIP-level racial composition and socioeconomic status, indicating that proxy errors are structured rather than random. This empirical pattern aligns with a growing methodological literature that enhances BISG and related approaches by augmenting them with additional contextual or socioeconomic features and applying machine-learning methods to reduce systematic misclassification. Recent examples include \citet{argyle2024misclassification}, who train random-forest models using BISG outputs and socioeconomic covariates to correct correlated proxy errors, and \citet{kwegyir2024observing}, who develop contextual proxy models (cBISG and MICSG) that incorporate outcome-specific context and supervised learning to achieve calibrated disparity estimates. Other work focuses on improving proxy calibration directly; \citet{greengard2023improved}, for example, correct BISG’s minority underclassification through raking-based adjustments. A related thread treats proxy race as a measurement-error problem and derives partial-identification bounds for disparity estimation \citep[e.g.,][]{kallus2022assessing}.

\section{Conclusion} \label{sec:conclusion}

This paper examines a setting that is becoming increasingly common in regulatory and industry practice but remains poorly understood in the academic literature: regression-based fairness assessments conducted without access to observed protected-class information and instead relying on proxy-imputed race. Our analysis shows that proxy race is not merely a noisy substitute for true race. When used as a categorical regressor, proxy-imputed race fundamentally alters the behavior of regression-based disparity estimates, producing coefficients that reflect weighted mixtures of group effects driven by the structure of misclassification. This mixing mechanism systematically shrinks estimated disparities toward the majority group--even when overall classification accuracy is high.

Our empirical analysis illustrates how these mechanisms operate in practice. Using a linked North Carolina voter–insurance dataset, we show that BIFSG misclassification is highly structured, varying systematically with ZIP-level racial composition and socioeconomic conditions. These structured errors interact with regression residuals in ways that can either attenuate or amplify conditional disparity estimates. Together, the theoretical and empirical results demonstrate that proxy-based regression is not a neutral substitute for analyses based on observed protected-class labels.

These findings have important implications for quantitative fairness testing in regulatory and governance settings. In emerging insurance frameworks such as Colorado’s SB21-169 draft testing framework and New York’s Circular Letter No. 7, proxy-based inference is often viewed as a practical workaround for the absence of protected-class data. Our results show that, without careful attention to the structure of misclassification and its interaction with model residuals, proxy-based audits risk either masking meaningful disparities--creating false reassurance through attenuated estimates--or generating spurious signals of unfairness driven by residual correlations rather than underlying discrimination. Accordingly, proxy-based testing should be interpreted as a diagnostic tool rather than a definitive measure of discrimination. 

Several directions for future research follow naturally from this work. First, our empirical findings suggest scope for audit methods that explicitly model or correct correlated proxy errors arising from geographic and socioeconomic structure. Second, complementary institutional approaches--such as trusted third-party audits, privacy-preserving data linkage, or partial-identification frameworks--may offer alternatives to unstructured reliance on proxy inference. Finally, extending the analysis to other protected attributes and application domains would help clarify the generality of the mechanisms documented here.

\bibliographystyle{apacite}
\bibliography{reference}

\newpage

\appendix
\section{Regulatory Context: Proxy-Based Fairness Testing in Insurance}
\label{app:regcontext}

This appendix summarizes the regulatory developments that motivate proxy-based race inference and regression-based disparity testing in insurance.

\subsection{Colorado SB21-169 and Draft Quantitative Testing Framework}

In July 2021, Colorado enacted Senate Bill 21-169 (SB21-169), a landmark insurance reform aimed at strengthening consumer protections against unfair discrimination in insurance practices. Colorado’s SB21-169 prohibits insurers from using external consumer data and information sources (ECDIS) and related algorithms or predictive models in any insurance practice in a manner that unfairly discriminates on the basis of protected classes, and holds insurers accountable for ``testing their big data systems'' with an expectation of corrective action where consumer harms are identified. In September 2023, the Colorado Division of Insurance released a draft quantitative testing regulation for life insurance underwriting pursuant to SB21-169. The draft proposes a regression-based framework under which insurers would be required to conduct annual quantitative testing to assess whether the use of ECDIS, or models relying on ECDIS, results in unfairly discriminatory underwriting or pricing outcomes based on the race or ethnicity of proposed insureds.

Under the 2023 draft, the testing procedure proceeds in several steps: (1) identify specified underwriting outcomes of interest, including underwriting decisions and premium outcomes; (2) infer the proposed insureds’ race or ethnicity using the Bayesian Improved First Name Surname Geocoding (BIFSG) method \citep{voicu2018using} based on individuals’ names and geolocation information; (3) estimate regression models in which the underwriting outcome is regressed on indicators for inferred race or ethnicity, adjusting for a specified set of control variables (policy type, face amount, age, gender, and tobacco use). A statistically significant disparity (p-value $<$ 0.05) that exceeds predefined magnitude thresholds for approval or premium outcomes triggers further variable-level testing to identify contributing ECDIS inputs. In contrast to the fair-lending context -- where proxy race measures were typically deployed by regulators in ex post regulatory investigations -- Colorado’s framework moves toward an ex ante, self-audit paradigm in which insurers are expected to implement internal governance and ongoing quantitative testing processes. However, industry stakeholders raised concerns that the draft’s highly prescriptive structure limited flexibility in the specification of regression-based testing procedures\footnote{Subsequent to the Colorado Division of Insurance’s September 2023 draft, industry stakeholders continued to engage with the Division through a formal stakeholder process. In June 2024, the American Council of Life Insurers (ACLI) presented an industry-proposed alternative draft quantitative testing regulation. The ACLI proposal differs in several material respects from the 2023 Division draft. In particular, it relaxes the prescriptive structure of the original testing framework by allowing greater flexibility in the choice of control variables and outcomes of interest, permits the use of BIFSG, BISG, or other approved inference methods, and reframes the regression analysis to condition explicitly on model outputs. This shift reflects a change in testing intention: whereas the 2023 Division draft emphasizes detecting outcome-level disparities associated with inferred race or ethnicity, the ACLI proposal focuses more narrowly on identifying residual disparities not already embodied in model outputs.}.

\subsection{New York Insurance Circular Letter No.\ 7}


In July 2024, the New York Department of Financial Services (NYDFS) issued Insurance Circular Letter No. 7, setting supervisory expectations for insurers’ use of ECDIS and artificial intelligence systems (AIS) in underwriting and pricing. In particular, NYDFS encourages insurers to conduct “quantitative assessments” of whether underwriting or pricing outcomes exhibit disproportionate adverse effects for similarly situated insureds or for insureds of a protected class, while clarifying that the Department does not expect insurers to collect additional protected-class data solely for this purpose. Although the Circular Letter does not mandate any specific statistical methodology, it explicitly contemplates that protected-class status may be determined using data available to the insurer or “reasonably inferred using accepted statistical methodologies.” This language creates a regulatory setting in which proxy-based approaches to race or ethnicity inference -- such as BIFSG -- may serve as practical tools for operationalizing quantitative testing when direct race information is unavailable. In contrast to Colorado’s more prescriptive draft framework, New York’s approach places greater discretion on insurers in how quantitative assessments are conducted, while holding insurers accountable for identifying and mitigating unfairly discriminatory outcomes. 

\subsection{Broader Policy Motivation and Demographic Data Constraints}

Beyond specific state-level insurance frameworks, a broader policy motivation underlies the increasing attention to proxy-based fairness testing: regulators and policymakers are seeking ways to assess unfair discrimination in settings where protected-class data are unavailable or legally constrained. In this context, the National Association of Insurance Commissioners (NAIC, \citeyear{NAIC2023AIModelBulletin}) adopted a model bulletin on insurers’ use of AI systems that emphasizes the need for insurers to implement appropriate testing and validation procedures to identify errors, bias, and the potential for unfair discrimination in the decisions and outcomes arising from predictive models and AI systems.

The challenges in accessing demographic data for fairness assessment also appear beyond the U.S. insurance or fair-lending setting. For example, UK government policy work on AI fairness \citep{UKGov2023DemographicDataFairer} identifies the lack of reliable demographic information as a major barrier to detecting and mitigating bias in AI systems, and explicitly considers proxies as a potential solution when direct collection of protected attributes (including race and sex) is infeasible. Similar constraints have also been encountered in private-sector settings, such as Airbnb’s Project Lighthouse, which relies on proxy-based and privacy-preserving approaches to assess potential discrimination without direct collection of protected attributes \citep{Airbnb2020Lighthouse}.

\section{Datasets Used}
\label{app:datasets-used}


\resizebox{1.1\textwidth}{!}{
\begin{lstlisting}
'data.frame':	1851858 obs. of  13 variables:
 $ zip_code       : Factor w/ 100 levels "27012","27028",..: 1 1 1 1 1 1 1 1 1 1 ...
 $ surname        : chr  "HILTIBIDAL" "FRANKLIN" "MILLER" "BEESON" ...
 $ first          : chr  "JULIA" "DENNY" "CARSON" "DEBORAH" ...
 $ race_code      : Factor w/ 8 levels "A","B","I","M",..: 8 8 8 8 8 8 8 8 8 8 ...
 $ ethnic_code    : Factor w/ 3 levels "HL","NL","UN": 2 2 3 3 2 2 2 3 2 2 ...
 $ gender_code    : Factor w/ 2 levels "F","M": 1 2 2 1 1 2 1 1 1 2 ...
 $ age_at_year_end: int  46 45 21 71 64 56 68 22 70 73 ...
 $ drivers_lic    : Factor w/ 1 level "Y": 1 1 1 1 1 1 1 1 1 1 ...
 $ age_range      : Factor w/ 18 levels "18","19","20",..: 12 12 4 17 15 14 16 5 17 17 ...
 $ average_premium: num  31 31 36 32 30 31 31 35 32 32 ...
 $ MEDFAMINC      : num  121863 121863 121863 121863 121863 ...
 $ PPOV           : num  0.0801 0.0801 0.0801 0.0801 0.0801 ...
 $ PUNEMP         : num  0.0601 0.0601 0.0601 0.0601 0.0601 ...
\end{lstlisting}
}

\subsection{Data Sources and Sample Construction}
\label{app:sample-construction}

The empirical dataset used in this study is constructed by merging three data sources: 
(i) individual-level voter registration records from the North Carolina State Board of Elections; 
(ii) ZIP-code–level auto insurance premiums obtained via web scraping from \textit{CarInsurance.com}; and 
(iii) ZIP Code Tabulation Area (ZCTA)–level socioeconomic indicators from the National Neighborhood Data Archive (NaNDA).

We begin with the North Carolina voter file and restrict the analysis to individuals residing in the first 100 ZIP codes in the state ranked by voter population. From the original voter data, we apply the following sample restrictions: 
(1) we retain only individuals located in the selected ZIP codes; 
(2) we remove observations for which both \texttt{race\_code} and \texttt{ethnic\_code} are recorded as “U” (undesignated), or for which \texttt{gender\_code} is “U”; and 
(3) we retain only individuals who hold a valid driver’s license.

Insurance premiums are then assigned to each individual by matching the average premium corresponding to their ZIP code $\times$ age group $\times$ gender cell. The premium data report single-year ages from 16 to 24 and, for older individuals, selected representative ages (25, 30, 35, \dots, 75). For individuals above age 24, voter ages are grouped into five-year intervals (25–29, 30–34, \dots, 70–74) and assigned the premium corresponding to the lower bound of each interval. Individuals older than 75 are excluded because premiums are not reported beyond that age. Since the voter file includes only individuals aged 18 or older, premium values for ages 16–17 are not used in the analysis. Therefore, these premiums are based on standardized driver profiles and vary across individuals only through demographic and geographic characteristics. We subsequently merge ZIP-level socioeconomic indicators from NaNDA using the individual’s residential ZIP code.

For empirical analysis, racial and ethnic information from the voter file is consolidated into mutually exclusive categories. Individuals with ethnicity code HL (Hispanic or Latino) are classified as Hispanic regardless of reported race. Among non-Hispanic individuals, race is grouped into Asian (A), Black (B), White (W), and Others (all remaining categories).

The final merged dataset contains 1,851,858 individual observations. 

\subsection{Variables Used in the Analysis}
\label{app:variables-used}

Table~\ref{tab:variables_ncvoter} lists the variables used in the empirical analysis reported in Section~\ref{sec:empirical}.

\begin{table}[H]
\centering
\begin{tabular}{lp{0.75\textwidth}}
\toprule
\textbf{Variable Name} & \textbf{Description} \\
\midrule
\texttt{zip\_code} & Residential ZIP code of the individual. 
Only the first 100 ZIP codes (ranked by voter population) are retained in the analysis. \\
\texttt{surname} &  Individual last name, used for proxy-based race inference. \\
\texttt{first} &  Individual first name, used for proxy-based race inference.\\
\texttt{race\_code} &  Self-reported race from the NC voter file. Categories include:
A (Asian),
B (Black or African American),
I (American Indian or Alaska Native),
M (Two or more races),
O (Other),
P (Native Hawaiian or Pacific Islander),
U (Undesignated),
W (White).\\
\texttt{ethnic\_code} &  Self-reported ethnicity from the NC voter file. Categories include: 
HL (Hispanic or Latino), 
NL (Not Hispanic or Not Latino), 
UN (Undesignated).\\
\texttt{gender\_code} & Gender/sex recorded in the voter file (F = Female, M = Male). 
Observations coded as U (Undesignated) are excluded from the analysis. \\
\texttt{age\_at\_year\_end} &  Age of the individual at the end of the calendar year.\\
\texttt{drivers\_lic} &  Indicator for holding a driver’s license (Y/N). 
Only individuals with a valid driver’s license (Y) are retained in the analytical sample. \\
\texttt{age\_range} &  Constructed categorical age variable with 18 effective age groups: 
18, 19, 20, 21, 22, 23, 24, 
25–29, 30–34, 35–39, 40–44, 45–49, 
50–54, 55–59, 60–64, 65–69, 70–74, and 75. 
Because the voter file includes only individuals aged 18 or older, younger ages do not appear in the data. Observations above age 75 are excluded for analysis. \\
\texttt{average\_premium} & Assigned auto insurance premium for each individual, determined by the ZIP code $\times$ age group $\times$ gender cell. \\
\texttt{MEDFAMINC} & Median family income at the ZIP Code Tabulation Area (ZCTA) level (2018–2022 ACS five-year estimates). \\
\texttt{PPOV} &  Proportion of individuals with income below the federal poverty level in the past 12 months at the ZCTA level (2018–2022 ACS five-year estimates).\\
\texttt{PUNEMP} &  Proportion of the civilian labor force (age 16+) unemployed at the ZCTA level (2018–2022 ACS five-year estimates). \\
\bottomrule
\end{tabular}
\caption{Variables Used in the Analysis}
\label{tab:variables_ncvoter}
\end{table}

\section{Connection to Classical Continuous Measurement Error}
\label{app:continuous-me}

This section contrasts our categorical-misclassification decomposition with the classical linear regression setting where covariates are continuous and subject to additive measurement error. Consider the standard model
\[
y = X\beta + \epsilon,\qquad \tilde X = X + E,
\]
where $E$ is a continuous-valued error matrix satisfying $\mathbb{E}[E\mid X]=0$ and (classically) $E \perp X$. The OLS estimator using $\tilde X$ is
\[
\tilde\beta = (\tilde X^\top \tilde X)^{-1}\tilde X^\top y
= (\tilde X^\top \tilde X)^{-1}\tilde X^\top (X\beta+\epsilon)
= (\tilde X^\top \tilde X)^{-1}\tilde X^\top X\,\beta
+ (\tilde X^\top \tilde X)^{-1}\tilde X^\top\epsilon.
\]
Taking expectations and using $\mathbb{E}[\epsilon\mid X,E]=0$,
\[
\mathbb{E}[\tilde\beta\mid X]
=
\mathbb{E}\!\left[(\tilde X^\top \tilde X)^{-1}\tilde X^\top X\mid X\right]\beta.
\]
Under the classical assumptions $E\perp X$ and $\mathbb{E}[E\mid X]=0$, we have
\[
\mathbb{E}[\tilde X^\top X\mid X] = X^\top X,
\qquad
\mathbb{E}[\tilde X^\top \tilde X\mid X] = X^\top X+\Sigma,
\]
where $\Sigma=\mathbb{E}[E^\top E]$. Applying a ratio-of-expectations approximation then yields
\[
\mathbb{E}[\tilde\beta] \approx (X^\top X+\Sigma)^{-1}X^\top X\,\beta
= \Big(I-(X^\top X+\Sigma)^{-1}\Sigma\Big)\beta,
\]
so the attenuation bias is
\[
\beta-\mathbb{E}[\tilde\beta] \approx (X^\top X+\Sigma)^{-1}\Sigma\,\beta,
\]
which corresponds to a ridge-type shrinkage intuition.

In contrast, for categorical misclassification, $E$ is structurally dependent on $X$. Under one-hot encoding with $E=\tilde X-X$, each row of $X$ and $\tilde X$ contains exactly one 1, and any misclassification produces a $-1$ entry at the true category and a $+1$ entry at the predicted category. In particular, $E_{ij}=-1$ can occur only when $X_{ij}=1$, so $E$ cannot be independent of $X$. Therefore, the classical independence assumptions and the resulting ridge-type attenuation approximation need not apply. This motivates the exact decomposition in Definition~\ref{def:decomposition} and the confusion-matrix representation developed in Section~\ref{sec:errors-categorical}.
\end{document}